\documentclass{article}%
\usepackage{amsmath}
\usepackage{amsfonts}
\usepackage{amssymb}
\usepackage{graphicx}
\usepackage[ruled,vlined,linesnumbered]{algorithm2e}%
\usepackage{tikz}
\providecommand{\U}[1]{\protect\rule{.1in}{.1in}}
\newtheorem{theorem}{Theorem} [section]

\newtheorem{corollary}[theorem]{Corollary}

\newtheorem{example}[theorem]{Example}

\newtheorem{lemma}[theorem]{Lemma}

\newenvironment{proof}[1][Proof]{\noindent\textbf{#1.} }{\ \rule{0.5em}{0.5em}}
\begin{document}

\author{Vadim E. Levit\\Department of Computer Science\\Ariel University, ISRAEL\\levitv@ariel.ac.il
\and David Tankus\\Department of Software Engineering\\Sami Shamoon College of Engineering, ISRAEL\\davidt@sce.ac.il}
\title{Complexity results for generating subgraphs}
\date{}
\maketitle

\begin{abstract}
A graph $G$ is \textit{well-covered} if all its maximal independent sets are
of the same cardinality. Assume that a weight function $w$ is defined on its
vertices. Then $G$ is $w$\textit{-well-covered} if all maximal independent
sets are of the same weight. For every graph $G$, the set of weight functions
$w$ such that $G$ is $w$-well-covered is a \textit{vector space}, denoted
$WCW(G)$.

Let $B$ be a complete bipartite induced subgraph of $G$ on vertex sets of
bipartition $B_{X}$ and $B_{Y}$. Then $B$ is \textit{generating} if there
exists an independent set $S$ such that $S \cup B_{X}$ and $S \cup B_{Y}$ are
both maximal independent sets of $G$. In the restricted case that a generating
subgraph $B$ is isomorphic to $K_{1,1}$, the unique edge in $B$ is called a
\textit{relating edge}.

Deciding whether an input graph $G$ is well-covered is \textbf{co-NP}%
-complete. Therefore finding $WCW(G)$ is \textbf{co-NP}-hard. Deciding whether
an edge is relating is \textbf{NP}-complete. Therefore, deciding whether a
subgraph is generating is \textbf{NP}-complete as well.

In this article we discuss the connections among these problems, provide
proofs for \textbf{NP}-completeness for several restricted cases, and present
polynomial characterizations for some other cases.\medskip

\textbf{Keywords:} weighted well-covered graph; maximal independent set;
relating edge; generating subgraph; vector space.

\end{abstract}

\section{Introduction}

\subsection{Basic definitions and notation}

Throughout this paper $G$ is a simple (i.e., a finite, undirected, loopless
and without multiple edges) graph with vertex set $V(G)$ and edge set $E(G)$.

Cycles of $k$ vertices are denoted by $C_{k}$. When we say that $G$ does not
contain $C_{k}$ for some $k \geq3$, we mean that $G$ does not admit subgraphs
isomorphic to $C_{k}$. Note that these subgraphs are not necessarily induced.
Let $\mathcal{G}(\widehat{C_{i_{1}}},..,\widehat{C_{i_{k}}})$ be the family of
all graphs which do not contain $C_{i_{1}}$,...,$C_{i_{k}}$.

Let $u$ and $v$ be two vertices in $G$. The \textit{distance} between $u$ and
$v$, denoted $d(u,v)$, is the length of a shortest path between $u$ and $v$,
where the length of a path is the number of its edges. If $S$ is a non-empty
set of vertices, then the \textit{distance} between $u$ and $S$, is defined as
$d(u,S)=\min\{d(u,s):s\in S\}$.

For every positive integer $i$, denote
\[
N_{i}(S)=\{x\in V\left(  G\right)  :d(x,S)=i\},
\]
and
\[
N_{i}\left[  S\right]  =\{x\in V\left(  G\right)  :d(x,S)\leq i\}.
\]

If $S$ contains a single vertex, $v$, then we abbreviate $N_{i}(\{v\})$,
$N_{i}\left[  \{v\}\right]  $ to be $N_{i}(v)$, $N_{i}\left[  v\right]  $,
respectively. We denote by $G[S]$ the subgraph of $G$ induced by $S$. For
every two sets, $S$ and $T$, of vertices of $G$, we say that $S$
\textit{dominates} $T$ if $T\subseteq N_{1}\left[  S\right]  $.

\subsection{Well-covered graphs}

Let $G$ be a graph. A set of vertices $S$ is \textit{independent} if its
elements are pairwise nonadjacent. An independent set of vertices is
\textit{maximal} if it is not a subset of another independent set. An
independent set of vertices is \textit{maximum} if the graph does not contain
an independent set of a higher cardinality.

The graph $G$ is \textit{well-covered} if every maximal independent set is
maximum \cite{plummer:definition}. Assume that a weight function $w:V\left(
G\right)  \longrightarrow\mathbb{R}$ is defined on the vertices of $G$. For
every set $S\subseteq V\left(  G\right)  $, define 
\[
w(S)={\sum\limits_{s\in S}}w(s). 
\]
Then $G$ is $w$\textit{-well-covered} if all maximal independent
sets of $G$ are of the same weight.

The problem of finding a maximum independent set is \textbf{NP-}complete.
However, if the input is restricted to well-covered graphs, then a maximum
independent set can be found in polynomial time using the \textit{greedy
algorithm}. Similarly, if a weight function $w:V\left(  G\right)
\longrightarrow\mathbb{R}$ is defined on the vertices of $G$, and $G$ is
$w$-well-covered, then finding a maximum weight independent set is a
polynomial problem. There is an interesting application, where well-covered
graphs are investigated in the context of distributed $k$-mutual exclusion
algorithms \cite{yaka:coteries}.

The recognition of well-covered graphs is known to be \textbf{co-NP}-complete.
This is proved independently in \cite{cs:note} and \cite{sknryn:compwc}. In
\cite{cst:structures} it is proven that the problem remains \textbf{co-NP}%
-complete even when the input is restricted to $K_{1,4}$-free graphs. However,
the problem can be solved in polynomial time for $K_{1,3}$-free graphs
\cite{tata:wck13f,tata:wck13fn}, for graphs with girth $5$ at least
\cite{fhn:wcg5}, for graphs with a bounded maximal degree \cite{cer:degree},
for chordal graphs \cite{ptv:chordal}, and for graphs without cycles of
lengths $4$ and $5$ \cite{fhn:wc45}.

For every graph $G$, the set of weight functions $w$ for which $G$ is
$w$-well-covered is a vector space \cite{cer:degree}. That vector space is
denoted $WCW(G)$ \cite{bnz:wcc4}. Since recognizing well-covered graphs is
\textbf{co-NP}-complete, finding the vector space $WCW(G)$ of an input graph
$G$ is \textbf{co-NP}-hard. However, finding $WCW(G)$ can be done in
polynomial time when the input is restricted to graphs with a bounded maximal
degree \cite{cer:degree}, to graphs without cycles of lengths $4$, $5$ and $6$
\cite{lt:wwc456}, and to chordal graphs \cite{bn:wcchordal}.

\subsection{Generating subgraphs and relating edges}

Further we make use of the following notions, which have been introduced in
\cite{lt:wc4567}. Let $B$ be an induced complete bipartite subgraph of $G$ on
vertex sets of bipartition $B_{X}$ and $B_{Y}$. Assume that there exists an
independent set $S$ such that each of $S\cup B_{X}$ and $S\cup B_{Y}$ is a
maximal independent set of $G$. Then $B$ is a \textit{generating subgraph} of
$G$, and the set $S$ is a \textit{witness} that $B$ is generating. We observe
that every weight function $w$ such that $G$ is $w$-well-covered must
\textit{satisfy} the restriction $w(B_{X})=w(B_{Y})$.

If the generating subgraph $B$ contains only one edge, say $xy$, it is called
a \textit{relating edge}. In such a case, the equality $w(x)=w(y)$ is valid
for every weight function $w$ such that $G$ is $w$-well-covered.

Recognizing relating edges is known to be \textbf{NP-}complete \cite{bnz:wcc4}%
, and it remains \textbf{NP-}complete even when the input is restricted to
graphs without cycles of lengths $4$ and $5$ \cite{lt:relatedc4}. Therefore,
recognizing generating subgraphs is also \textbf{NP-}complete when the input
is restricted to graphs without cycles of lengths $4$ and $5$. However,
recognizing relating edges can be done in polynomial time if the input is
restricted to graphs without cycles of lengths $4$ and $6$ \cite{lt:relatedc4}%
, and to graphs without cycles of lengths $5$ and $6$ \cite{lt:wwc456}.

It is also known that recognizing generating subgraphs is a polynomial problem
when the input is restricted to graphs without cycles of lengths $4$, $6$ and
$7$ \cite{lt:wc4567}, to graphs without cycles of lengths $4$, $5$ and $6$
\cite{lt:wwc456}, and to graphs without cycles of lengths $5$, $6$ and $7$
\cite{lt:wwc456}.

\subsection{Introducing the problems under consideration}

The subject of this article is the following four problems and their interconnections.

\begin{itemize}
\item $\mathbf{WC}$\textbf{ problem}:\newline\textit{Input}: A graph
$G$.\newline\textit{Question}: Is $G$ well-covered?

\item $\mathbf{WCW}$\textbf{ problem}:\newline\textit{Input}: A graph
$G$.\newline\textit{Output}: The vector space $WCW(G)$.

\item $\mathbf{GS}$ \textbf{problem}:\newline\textit{Input}: A graph $G$, and
an induced complete bipartite subgraph $B$ of $G$.\newline\textit{Question}:
Is $B$ generating?

\item $\mathbf{RE}$\textbf{ problem}:\newline\textit{Input}: A graph $G$, and
an edge $xy\in E\left(  G\right)  $.\newline\textit{Question}: Is $xy$
relating?\newline
\end{itemize}

If we know the output of the $\mathbf{WCW}$ problem for a graph $G$, then we
know the output of the $\mathbf{WC}$ problem for the same $G$: The graph $G$
is well-covered if and only if $w\equiv1$ belongs to $WCW(G)$. Therefore, the
$\mathbf{WC}$ problem is not harder than the $\mathbf{WCW}$ problem. Let
$\Psi$ be a family of graphs. If the $\mathbf{WCW}$ problem can be solved in
polynomial time, when its input is restricted to $\Psi$, then also the
$\mathbf{WC}$ problem is polynomial, when its input is restricted to $\Psi$.
On the other hand, if the $\mathbf{WC}$ problem is \textbf{co}-\textbf{NP-}%
complete, when its input is restricted to $\Psi$, then the $\mathbf{WCW}$
problem is \textbf{co}-\textbf{NP-}hard, when its input is restricted to
$\Psi$.

A similar connection exists between the $\mathbf{GS}$ problem and the
$\mathbf{RE}$ problem, since the $\mathbf{RE}$ problem is a restricted case of
the $\mathbf{GS}$ problem. Therefore, for every family $\Psi$ of graphs, if
the $\mathbf{GS}$ problem can be solved in polynomial time, then the
$\mathbf{RE}$ problem can be solved in polynomial time as well, and if the
$\mathbf{RE}$ problem is \textbf{NP-}complete then the $\mathbf{GS}$ problem
is also \textbf{NP-}complete.

This article considers bipartite graphs, graphs with girth $6$ at least, and
$K_{1,4}$-free graphs. Although for bipartite graphs and graphs with girth $6$
at least, the $\mathbf{WC}$ problem is known to be solvable in polynomial
time, we prove that the $\mathbf{GS}$ problem is \textbf{NP-}complete. For
bipartite graphs, even the $\mathbf{RE}$ problem is \textbf{NP-}complete.
Additionally, \textbf{NP-}completeness of the $\mathbf{GS}$ problem for
$K_{1,4}$-free graphs is proved. We also present polynomial algorithms for the
$\mathbf{RE}$ problem, the $\mathbf{GS}$ problem, and the $\mathbf{WCW}$
problem in the case that the maximum degree of the input graph is bounded.

\section{\textbf{NP-}complete cases}

A \textit{binary variable} is a variable whose value is either $0$ or $1$. If
$x$ is a binary variable, then its \textit{negation} is denoted by
$\overline{x}$. Each of $x$ and $\overline{x}$ are called \textit{literals}.
Let $X=\{x_{1},...,x_{n}\}$ be a set of binary variables. A \textit{clause}
$c$ over $X$ is a set of literals belonging to $\{x_{1},\overline{x_{1}%
},...,x_{n},\overline{x_{n}}\}$ such that $c$ does not contain both a variable
and its negation. A \textit{truth assignment} is a function
\[
\Phi:\{x_{1},\overline{x_{1}},...,x_{n},\overline{x_{n}}\}\longrightarrow
\{0,1\}
\]
such that
\[
\Phi(\overline{x_{i}})=1-\Phi(x_{i})\text{ for each }1\leq i\leq n.
\]
A truth assignment $\Phi$ \textit{satisfies} a clause $c$ if $c$ contains at
least one literal $l$ such that $\Phi(l)=1$.

\subsection{Relating edges in bipartite graphs}

In this subsection we consider the following problems:

\begin{itemize}
\item \textbf{SAT problem}:\newline\textit{Input}: A set $X$ of binary
variables and a set $C$ of clauses over $X$.\newline\textit{Question}: Is
there a truth assignment for $X$ which satisfies all clauses of $C$?

\item \textbf{BWSAT problem}:\newline\textit{Input}: A set $X$ of binary
variables and two sets, $C_{1}$ and $C_{2}$, of clauses over $X$, such that
all literals of the clauses belonging to $C_{1}$ are variables, and all
literals of clauses belonging to $C_{2}$ are negations of variables.
\newline\textit{Question}: Is there a truth assignment for $X$, which
satisfies all clauses of $C_{1}\cup C_{2}$?\newline
\end{itemize}

By Cook-Levin's Theorem, the \textbf{SAT} problem is \textbf{NP-}complete. We
prove that the same holds for the \textbf{BWSAT }problem.

\begin{lemma}
\label{usat} The \textbf{BWSAT} problem is in \textbf{NP-}complete.
\end{lemma}

\begin{proof}
Obviously, the \textbf{BWSAT} problem is in \textbf{NP}. We prove its
\textbf{NP-}completeness by showing a reduction from the \textbf{SAT} problem.
Let
\[
I_{1}=(X=\{x_{1},...,x_{n}\},C=\{c_{1},...,c_{m}\})
\]
be an instance of the \textbf{SAT} problem. Define $Y=\{x_{1},...,x_{n}%
,y_{1},...,y_{n}\}$, where $y_{1},...,y_{n}$ are new variables. For every
$1\leq j\leq m$, let $c_{j}^{\prime}$ be the clause obtained from $c_{j}$ by
replacing $\overline{x_{i}}$ with $y_{i}$ for each $1\leq i\leq n$. Let
$C^{\prime}=\{c_{1}^{\prime},...,c_{m}^{\prime}\}$. For each $1\leq i\leq n$
define two new clauses, $d_{i}=\{x_{i},y_{i}\}$ and $e_{i}=\{\overline{x_{i}%
},\overline{y_{i}}\}$. Let $D=\{d_{1},...,d_{n}\}$ and $E=\{e_{1},...,e_{n}%
\}$. Obviously, all literals of $C^{\prime}\cup D$ are variables, and all
literals of $E$ are negations of variables. Hence, $I_{2}=(Y,C^{\prime}\cup
D,E)$ is an instance of the \textbf{BWSAT} problem, see Example \ref{example1}%
. It remains to prove that $I_{1}$ and $I_{2}$ are equivalent.

Assume that $I_{1}$ is a positive instance of the \textbf{SAT} problem. There
exists a truth assignment
\[
\Phi_{1}:\{x_{1},\overline{x_{1}},...,x_{n},\overline{x_{n}}\}\longrightarrow
\{0,1\}
\]
which satisfies all clauses of $C$. Extend $\Phi_{1}$ to a truth assignment
\[
\Phi_{2}:\{x_{1},\overline{x_{1}},...,x_{n},\overline{x_{n}},y_{1}%
,\overline{y_{1}},...,y_{n},\overline{y_{n}}\}\longrightarrow\{0,1\}
\]
by defining $\Phi_{2}(y_{i})=1-\Phi_{1}(x_{i})$ for each $1\leq i\leq n$.
Clearly, $\Phi_{2}$ is a truth assignment which satisfies all clauses of
$C^{\prime}\cup D\cup E$. Hence, $I_{2}$ is a positive instance of the
\textbf{BWSAT} problem.

Assume $I_{2}$ is a positive instance of the \textbf{BWSAT} problem. There
exists a truth assignment
\[
\Phi_{2}:\{x_{1},\overline{x_{1}},...,x_{n},\overline{x_{n}},y_{1}%
,\overline{y_{1}},...,y_{n},\overline{y_{n}}\}\longrightarrow\{0,1\}
\]
that satisfies all clauses of $C^{\prime}\cup D\cup E$. For every $1\leq i\leq
n$ it holds that $\Phi_{2}(y_{i})=1-\Phi_{2}(x_{i})$, or otherwise one of
$d_{i}$ and $e_{i}$ is not satisfied. Therefore, $I_{1}$ is a positive
instance of the \textbf{SAT} problem.
\end{proof}

\begin{example}
\label{example1} The following contains both an instance of the \textbf{SAT}
problem and an equivalent instance of the \textbf{BWSAT} problem.

$I_{1}=(X,C)$, where $X=\{x_{1},x_{2},x_{3},x_{4},x_{5}\}$,\newline%
$C=\{\{x_{1},\overline{x_{2}},x_{3}\},\{x_{1},x_{3},x_{4},x_{5}\},\{\overline
{x_{1}},x_{2},\overline{x_{3}},x_{4}\},\{x_{1},x_{2},\overline{x_{4}%
},\overline{x_{5}}\}\}$,\newline$I_{2}=(Y,C_{1},C_{2})$, where $Y=\{x_{1}%
,x_{2},x_{3},x_{4},x_{5},y_{1},y_{2},y_{3},y_{4},y_{5}\}$, \newline%
$C_{1}=\{\{x_{1},y_{2},x_{3}\},\{x_{1},x_{3},x_{4},x_{5}\},\{y_{1},x_{2}%
,y_{3},x_{4}\},\{x_{1},x_{2},y_{4},y_{5}\},\{x_{1},y_{1}\},\newline%
\{x_{2},y_{2}\},\{x_{3},y_{3}\},\{x_{4},y_{4}\},\{x_{5},y_{5}\}\}$.
\newline$C_{2}=\{\{\overline{x_{1}},\overline{y_{1}}\},\{\overline{x_{2}%
},\overline{y_{2}}\},\{\overline{x_{3}},\overline{y_{3}}\},\{\overline{x_{4}%
},\overline{y_{4}}\},\{\overline{x_{5}},\overline{y_{5}}\}\}$
\end{example}

The following theorem is the main result of this section.

\begin{theorem}
\label{bipartitenpc} The $\mathbf{RE}$ problem is \textbf{NP-}complete even if
its input is restricted to bipartite graphs.
\end{theorem}

\begin{proof}
The problem is obviously in \textbf{NP}. We prove \textbf{NP-}completeness by
showing a reduction from the \textbf{BWSAT} problem. Let
\[
I_{1}=(X=\{x_{1},...,x_{n}\},C_{1},C_{2})
\]
be an instance of the \textbf{BWSAT} problem, where $C_{1}=\{c_{1}%
,...,c_{m}\}$ is a set of clauses which contain only variables, and
$C_{2}=\{c_{1}^{\prime},...,c_{m^{\prime}}^{\prime}\}$ is a set of clauses
which contain only negations of variables. Define a graph $B$ as follows:
\begin{gather*}
V\left(  B\right)  =\{x,y\}\cup\{v_{j}:1\leq j\leq m\}\cup\{v_{j}^{\prime
}:1\leq j\leq m^{\prime}\}\cup\\
\{u_{i}:1\leq i\leq n\}\cup\{u_{i}^{\prime}:1\leq i\leq n\},
\end{gather*}%
\begin{gather*}
E\left(  B\right)  =\{xy\}\cup\{xv_{j}:1\leq j\leq m\}\cup\{yv_{j}^{\prime
}:1\leq j\leq m^{\prime}\}\cup\\
\{v_{j}u_{i}:x_{i}\ \text{appears\ in}\ c_{j}\}\cup\{v_{j}^{\prime}%
u_{i}^{\prime}:\overline{x_{i}}\ \text{appears\ in}\ c_{j}^{\prime}%
\}\cup\{u_{i}u_{i}^{\prime}:1\leq i\leq n\}\newline.
\end{gather*}
\newline Clearly, $B$ is bipartite, and the vertex sets of its bipartition
are
\[
\{u_{i}:1\leq i\leq n\}\cup\{x\}\cup\{v_{j}^{\prime}:1\leq j\leq m^{\prime}\}
\]
and
\[
\{v_{j}:1\leq j\leq m\}\cup\{y\}\cup\{u_{i}^{\prime}:1\leq i\leq n\}.
\]

Consider the instance $I_{2}=(B,xy)$ of the $\mathbf{RE}$ problem. It is
necessary to prove that $I_{1}$ and $I_{2}$ are equivalent.

Assume that $I_{1}$ is a positive instance of the \textbf{BWSAT} problem. Let
\[
\Phi:\{x_{1},\overline{x_{1}},...,x_{n},\overline{x_{n}}\}\longrightarrow
\{0,1\}
\]
be a truth assignment which satisfies all clauses of $C_{1}\cup C_{2}$. Let
\[
S=\{u_{i}:\Phi(x_{i})=1\}\cup\{u_{i}^{\prime}:\Phi(x_{i})=0\}.
\]
Obviously, $S$ is independent. Since $\Phi$ satisfies all clauses of
$C_{1}\cup C_{2}$, every vertex of
\[
\{v_{j}:1\leq j\leq m\}\cup\{v_{j}^{\prime}:1\leq j\leq m^{\prime}\}
\]
is adjacent to a vertex of $S$. Hence, $S\cup\{x\}$ and $S\cup\{y\}$ are
maximal independent sets. Therefore, $S$ is a witness that $xy$ is a relating
edge, and $I_{2}$ is a positive instance of the $\mathbf{RE}$ problem.

On the other hand, assume that $I_{2}$ is a positive instance of the
$\mathbf{RE}$ problem. Let $S$ be a witness of $xy$. Since $S$ is a maximal
independent set of
\[
\{u_{i}:1\leq i\leq n\}\cup\{u_{i}^{\prime}:1\leq i\leq n\},
\]
exactly one of $u_{i}$ and $u_{i}^{\prime}$ belongs to $S$, for every $1\leq
i\leq n$. Let
\[
\Phi:\{x_{1},\overline{x_{1}},...,x_{n},\overline{x_{n}}\}\longrightarrow
\{0,1\}
\]
be a truth assignment defined by: $\Phi(x_{i})=1\iff u_{i}\in S$. The fact
that $S$ dominates
\[
\{v_{j}:1\leq j\leq m\}\cup\{v_{j}^{\prime}:1\leq j\leq m^{\prime}\}
\]
implies that all clauses of $C_{1}\cup C_{2}$ are satisfied by $\Phi$.
Therefore, $I_{1}$ is a positive instance of the \textbf{BWSAT} problem.
\end{proof}

\begin{figure}[h]
\setlength{\unitlength}{1.0cm} \begin{picture}(20,10)\thicklines
\put(6,4){\circle*{0.1}}
\put(6,5.5){\circle*{0.1}}
\put(5.7,4){\makebox(0,0){$y$}}
\put(5.7,5.5){\makebox(0,0){$x$}}
\put(1.5,8.5){\circle*{0.25}}
\put(3,8.5){\circle*{0.1}}
\put(4.5,8.5){\circle*{0.25}}
\put(7.5,8.5){\circle*{0.25}}
\put(9,8.5){\circle*{0.25}}
\put(10.5,8.5){\circle*{0.1}}
\put(1.2,8.5){\makebox(0,0){$u_{1}$}}
\put(2.7,8.5){\makebox(0,0){$u_{2}$}}
\put(4.2,8.5){\makebox(0,0){$u_{3}$}}
\put(7.8,8.5){\makebox(0,0){$u_{4}$}}
\put(9.3,8.5){\makebox(0,0){$u_{5}$}}
\put(10.8,8.5){\makebox(0,0){$u_{6}$}}
\multiput(3,7)(1.5,0){5}{\circle*{0.1}}
\put(2.7,6.9){\makebox(0,0){$v_{1}$}}
\put(4.2,6.9){\makebox(0,0){$v_{2}$}}
\put(5.7,6.9){\makebox(0,0){$v_{3}$}}
\put(7.8,6.9){\makebox(0,0){$v_{4}$}}
\put(9.3,6.9){\makebox(0,0){$v_{5}$}}
\multiput(4.5,2.5)(1.5,0){3}{\circle*{0.1}}
\put(4.2,2.6){\makebox(0,0){$v'_{1}$}}
\put(5.7,2.6){\makebox(0,0){$v'_{2}$}}
\put(7.8,2.6){\makebox(0,0){$v'_{3}$}}
\put(1.5,1){\circle*{0.1}}
\put(3,1){\circle*{0.25}}
\put(4.5,1){\circle*{0.1}}
\put(7.5,1){\circle*{0.1}}
\put(9,1){\circle*{0.1}}
\put(10.5,1){\circle*{0.25}}
\put(1.2,1){\makebox(0,0){$u'_{1}$}}
\put(2.7,1){\makebox(0,0){$u'_{2}$}}
\put(4.2,1){\makebox(0,0){$u'_{3}$}}
\put(7.8,1){\makebox(0,0){$u'_{4}$}}
\put(9.3,1){\makebox(0,0){$u'_{5}$}}
\put(10.8,1){\makebox(0,0){$u'_{6}$}}
\multiput(3,7)(1.5,0){5}{\circle*{0.1}}
\multiput(4.5,2.5)(1.5,0){3}{\circle*{0.1}}
\put(6,2.5){\line(0,1){4.5}}
\put(6,5.5){\line(-2,1){3}}
\put(6,5.5){\line(-1,1){1.5}}
\put(6,5.5){\line(1,1){1.5}}
\put(6,5.5){\line(2,1){3}}
\put(6,4){\line(-1,-1){1.5}}
\put(6,4){\line(1,-1){1.5}}
\put(3,7){\line(-1,1){1.5}}
\put(3,7){\line(0,1){1.5}}
\put(3,7){\line(1,1){1.5}}
\put(4.5,7){\line(-1,1){1.5}}
\put(4.5,7){\line(2,1){3}}
\put(6,7){\line(-3,1){4.5}}
\put(6,7){\line(1,1){1.5}}
\put(7.5,7){\line(-4,1){6}}
\put(7.5,7){\line(1,1){1.5}}
\put(7.5,7){\line(2,1){3}}
\put(9,7){\line(-3,1){4.5}}
\put(9,7){\line(0,1){1.5}}
\put(9,7){\line(1,1){1.5}}
\put(4.5,2.5){\line(-1,-1){1.5}}
\put(4.5,2.5){\line(0,-1){1.5}}
\put(4.5,2.5){\line(-2,-1){3}}
\put(6,2.5){\line(1,-1){1.5}}
\put(6,2.5){\line(-1,-1){1.5}}
\put(6,2.5){\line(2,-1){3}}
\put(6,2.5){\line(-2,-1){3}}
\put(7.5,2.5){\line(1,-1){1.5}}
\put(7.5,2.5){\line(0,-1){1.5}}
\put(7.5,2.5){\line(-3,-1){4.5}}
\put(7.5,2.5){\line(2,-1){3}}
\put(1.5,1){\line(0,1){7.5}}
\put(2.3,4.8){\oval(3,9.4)[l]}
\put(1.9,8.2){\oval(2.2,2.6)[t]}
\put(1.9,1){\oval(2.2,1.8)[b]}
\put(3.2,4.8){\oval(5.2,9.8)[l]}
\put(2.55,8.5){\oval(3.9,2.4)[t]}
\put(2.55,1.8){\oval(3.9,3.8)[b]}
\put(10.5,1){\line(0,1){7.5}}
\put(9.7,4.8){\oval(3,9.4)[r]}
\put(10.1,8.2){\oval(2.2,2.6)[t]}
\put(10.1,1){\oval(2.2,1.8)[b]}
\put(8.8,4.8){\oval(5.2,9.8)[r]}
\put(9.45,8.5){\oval(3.9,2.4)[t]}
\put(9.45,1.8){\oval(3.9,3.8)[b]}
\end{picture}
\caption{An example of the reduction from the \textbf{USAT} problem to the
$\mathbf{RE}$ problem.}%
\label{Fig1}%
\end{figure}
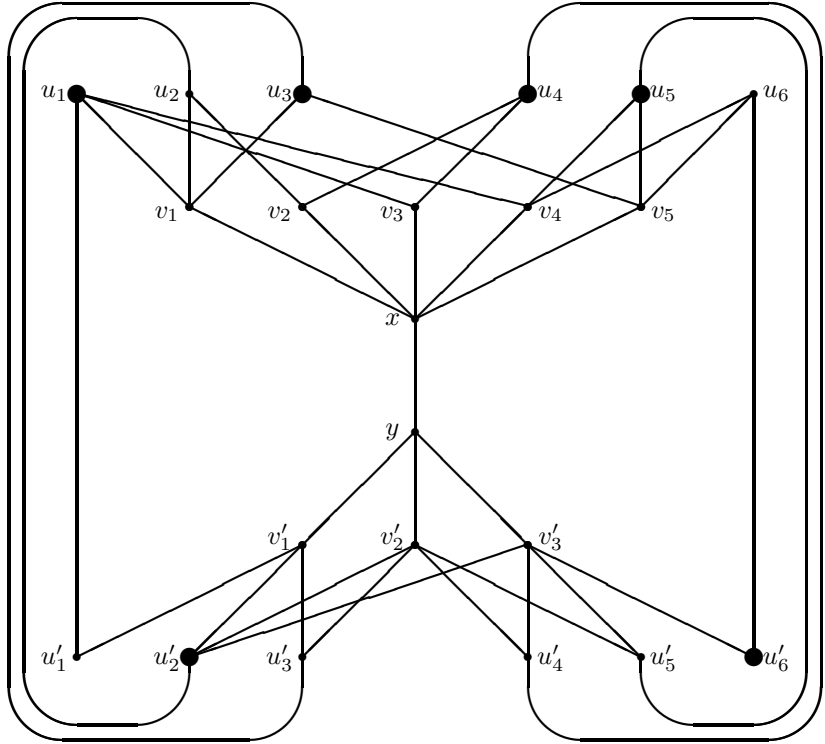

\begin{example}
\label{rereduction} Let $I_{1}=(X,C_{1},C_{2})$ be an instance of the
\textbf{BWSAT} problem, where $X=\{x_{1},x_{2},x_{3},x_{4},x_{5},x_{6}%
\}$,\newline$C_{1}=\{\{x_{1},x_{2},x_{3}\},\{x_{2},x_{4}\},\{x_{1}%
,x_{4}\},\{x_{1},x_{5},x_{6}\},\{x_{3},x_{5},x_{6}\}\}$, and \newline%
$C_{2}=\{\{\overline{x_{1}},\overline{x_{2}},\overline{x_{3}}\},\{\overline
{x_{2}},\overline{x_{3}},\overline{x_{4}},\overline{x_{5}}\},\{\overline
{x_{2}},\overline{x_{4}},\overline{x_{5}},\overline{x_{6}}\}\}$

Then $I_{2}=(G,xy)$ is an equivalent instance of the $\mathbf{RE}$ problem,
where $G$ is the graph shown in Figure \ref{Fig1}. The instance $I_{1}$ is
positive because of the satisfying assignment $\Phi$ defined by $\Phi(x_{i}) =
0$ if $i \in\{2,6\}$, and $\Phi(x_{i}) = 1$ otherwise. The corresponding
witness that $I_{2}$ is positive is the set $\{u_{1},u_{2}^{\prime}%
,u_{3},u_{4},u_{5},u_{6}^{\prime}\}$.
\end{example}

\begin{corollary}
\label{bipartitenpc2} The $\mathbf{GS}$ problem is \textbf{NP-}complete when
its input is restricted to bipartite graphs.
\end{corollary}

\subsection{Graphs with girth $6$ at least}

In this subsection we consider the following problems:

\begin{itemize}
\item \textbf{3-SAT problem}:\newline\textit{Input}: A set $X$ of binary
variables and a set $C$ of clauses over $X$ such that every clause contains
exactly $3$ literals.\newline\textit{Question}: Is there a truth assignment
for $X$ satisfying all clauses of $C$?

\item \textbf{DSAT problem}:\newline\textit{Input}: A set $X$ of binary
variables and a set $C$ of clauses over $X$ such that the following holds:

\begin{itemize}
\item Every clause contains $2$ or $3$ literals.

\item Every two distinct clauses have at most one literal in common.

\item If two clauses, $c_{1}$ and $c_{2}$, have a common literal $l_{1}$, then
there does not exist a literal $l_{2}$ such that $c_{1}$ contains $l_{2}$ and
$c_{2}$ contains $\overline{l_{2}}$.
\end{itemize}

\textit{Question}: Is there a truth assignment for $X$ satisfying all clauses
of $C$?
\end{itemize}

Let $I=(X, C)$ be an instance of the \textbf{3-SAT} problem. A \textit{bad
pair of clauses} is a set of two clauses $\{c_{1}, c_{2}\} \subseteq C$ such
that there exist literals, $l_{1}, l_{2}, l_{3}, l_{4}, l_{5}$, and:

\begin{itemize}
\item $c_{1}=\{l_{1},l_{2},l_{3}\}$ and $c_{2}=\{l_{1},l_{4},l_{5}\}$;

\item either $l_{2}=l_{4}$ or $l_{2}=\overline{l_{4}}$.
\end{itemize}

Clearly, an instance of the \textbf{3-SAT} problem with no bad pair of clauses
is also an instance of the \textbf{DSAT} problem. The \textbf{3-SAT} problem
is known to be in \textbf{NP-}complete. We prove that the same holds for the
\textbf{DSAT} problem.

\begin{lemma}
\label{dsat} The \textbf{DSAT} problem is \textbf{NP-}complete.
\end{lemma}

\begin{proof}
Obviously, the \textbf{DSAT} problem is in \textbf{NP}. We prove its
\textbf{NP-}completeness by showing a reduction from the \textbf{3-SAT}
problem. Let
\[
I_{1}=(X=\{x_{1},...,x_{n}\},C=\{c_{1},...,c_{m}\})
\]
be an instance of the \textbf{3-SAT} problem.

Assume that there exists a bad pair of clauses, $\{c_{j_{1}}, c_{j_{2}}\}
\subseteq C$, i.e. there exist literals, $l_{1}, l_{2}, l_{3}, l_{4}, l_{5}$,
such that:

\begin{itemize}
\item $c_{j_{1}}=\{l_{1},l_{2},l_{3}\}$ and $c_{j_{2}}=\{l_{1},l_{4},l_{5}\}$;

\item either $l_{2}=l_{4}$ or $l_{2}=\overline{l_{4}}$.
\end{itemize}

Define a new binary variable $x_{n+1}$, and new clauses $c_{j_{2}}^{1}%
=\{l_{1},x_{n+1},l_{5}\}$, $c_{j_{2}}^{2}=\{\overline{l_{4}},x_{n+1}\}$, and
$c_{j_{2}}^{3}=\{l_{4},\overline{x_{n+1}}\}$. Then
\[
I_{1}^{\prime}=(X\cup\{x_{n+1}\},(C\setminus\{c_{j_{2}}\})\cup\{c_{j_{2}}%
^{1},c_{j_{2}}^{2},c_{j_{2}}^{3}\})
\]
is an instance of the \textbf{SAT} problem.

We prove that $I_{1}$, $I_{1}^{\prime}$ are equivalent. Assume that $I_{1}$ is
a positive instance of the \textbf{3-SAT} problem. There exists a truth
assignment
\[
\Phi_{1}:\{x_{1},\overline{x_{1}},...,x_{n},\overline{x_{n}}\}\longrightarrow
\{0,1\}
\]
which satisfies all clauses of $C$. Extend $\Phi_{1}$ to a truth assignment
\[
\Phi_{2}:\{x_{1},\overline{x_{1}},...,x_{n+1},\overline{x_{n+1}}%
\}\longrightarrow\{0,1\}
\]
by defining $\Phi_{2}(x_{n+1})=\Phi_{1}(l_{4})$. Clearly, $\Phi_{2}$ satisfies
all clauses of $I_{1}^{\prime}$. On the other hand, assume that there exists a
truth assignment
\[
\Phi_{2}:\{x_{1},\overline{x_{1}},...,x_{n+1},\overline{x_{n+1}}%
\}\longrightarrow\{0,1\}
\]
which satisfies all clauses of $I_{1}^{\prime}$. Clauses $c_{j_{2}}^{2}$, and
$c_{j_{2}}^{3}$ imply that $\Phi_{2}(x_{n+1})=\Phi_{2}(l_{4})$. Therefore,
$I_{1}$ is a positive instance of the \textbf{3-SAT} problem.

The new clauses we added contain a new binary variable. Hence, they do not
belong to bad pairs of clauses. Moreover, the clause $c_{j_{2}}$ which belongs
to a bad pair in $I_{1}$ was omitted in $I_{1}^{\prime}$. Hence, the number of
bad pairs of clauses in $I_{1}^{\prime}$ is smaller than the one in $I_{1}$.

Repeat that process until an instance without bad pairs of clauses is
obtained, and denote that instance $I_{2}$. Clearly, every clause of $I_{2}$
has $2$ or $3$ literals. Hence, $I_{2}$ is an instance of the \textbf{DSAT}
problem, and $I_{1}$ and $I_{2}$ are equivalent.
\end{proof}

\begin{example}
\label{ex23sat} The following contains an instance of the \textbf{3-SAT}
problem and an equivalent instance of the \textbf{DSAT} problem.

$I_{1} = (X_{1}, C_{1})$ where $X_{1} = \{x_{1}, x_{2}, x_{3}, x_{4}, x_{5}\}$
and $C_{1} = \{\{x_{1}, \overline{x_{2}}, x_{3}\}, \newline\{x_{1}, x_{3},
x_{4}\}, \{x_{1}, x_{3}, x_{5}\}, \{\overline{x_{3}}, \overline{x_{4}},
x_{5}\}, \{x_{2}, \overline{x_{3}},\overline{x_{4}}\}, \{\overline{x_{1}},
\overline{x_{4}}, \overline{x_{5}}\}\}$.

$I_{2} = (X_{2}, C_{2})$ where $X_{2} = \{x_{1}, x_{2}, x_{3}, x_{4}, x_{5},
z_{3}, y_{3}, y_{4}, y_{5}\}$ and \newline$C_{2} = \{ \{x_{1},\overline{x_{2}%
},x_{3}\}, \{x_{1},y_{3},x_{4}\}, \{x_{1},z_{3},x_{5}\}, \{\overline{x_{3}%
},\overline{x_{4}},x_{5}\}, \{x_{2},\overline{x_{3}},y_{4}\}, \newline%
\{\overline{x_{1}},\overline{x_{4}},y_{5}\}, \{\overline{x_{3}},y_{3}\},
\{x_{3},\overline{y_{3}}\}, \{\overline{x_{3}},z_{3}\}, \{x_{3},\overline
{z_{3}}\}, \{x_{4},y_{4}\}, \{\overline{x_{4}},\overline{y_{4}}\},
\{x_{5},y_{5}\}, \newline\{\overline{x_{5}},\overline{y_{5}}\}\}$.
\end{example}

\begin{theorem}
\label{generatingc345npc} The following problem is \textbf{NP-}complete:
\newline Input: A graph $G\in\mathcal{G}(\widehat{C_{3}},\widehat{C_{4}%
},\widehat{C_{5}})$ and an induced complete bipartite subgraph $B$ of $G$.
\newline Question: Is $B$ generating?
\end{theorem}

\begin{proof}
The problem is obviously in \textbf{NP}. We prove its \textbf{NP-}completeness
by showing a reduction from the \textbf{DSAT} problem. Let
\[
I=(X=\{x_{1},...,x_{n}\},C=\{c_{1},...,c_{m}\})
\]
be an instance of the \textbf{DSAT} problem. Define a graph $G$ as follows.
\begin{align*}
V(G)  &  =\{y\}\cup\{a_{j}:1\leq j\leq m\}\cup\{v_{j}:1\leq j\leq m\}\cup\\
\{u_{i}  &  :1\leq i\leq n\}\cup\{u_{i}^{\prime}:1\leq i\leq n\}.
\end{align*}%
\begin{gather*}
E(G)=\{ya_{j}:1\leq j\leq m\}\cup\{a_{j}v_{j}:1\leq j\leq m\}\cup\{v_{j}%
u_{i}:x_{i}\ \text{appears\ in}\ c_{j}\}\cup\\
\{v_{j}u_{i}^{\prime}:\overline{x_{i}}\ \text{appears\ in}\ c_{j}\}\cup
\{u_{i}u_{i}^{\prime}:1\leq i\leq n\}.
\end{gather*}
Since a clause can not contain both a variable and its negation, \ $G$ \ does
\ not \ contain \ $C_{3}$. \ The \ fact \ that there are no pairs of bad
clauses implies that \ $G$ \ \ does \ \ not \ contain \ $C_{4}$ \ and
\ $C_{5}$. \ Hence, \ \ $G\in\mathcal{G}(\widehat{C_{3}},\widehat{C_{4}%
},\widehat{C_{5}})$. \ Let \ \ $B=G[\{y\}\cup\{a_{j}:1\leq j\leq m\}]$.
Obviously, $B$ is complete bipartite. Then $J=(G,B)$ is an instance of the
$\mathbf{GS}$ problem. It remains to prove that $I$ and $J$ are equivalent.

Assume that $I$ is positive, and let
\[
\Phi:\{x_{1},\overline{x_{1}},...,x_{n},\overline{x_{n}}\}\longrightarrow
\{0,1\}
\]
be a truth assignment which satisfies all clauses of $C$. Define
\[
S=\{u_{i}:\Phi(x_{i})=1\}\cup\{u_{i}^{\prime}:\Phi(x_{i})=0\}.
\]
Obviously, $S$ is independent. Since $\Phi$ satisfies all clauses of $C$, the
set $S$ dominates $\{v_{j}:1\leq j\leq m\}$. Hence, $S$ is a witness that $B$
is generating, i.e., $J$ is positive.

Assume that $J$ is positive. Let $S$ be a witness that $B$ is generating, and
let $S^{\ast}$ be a maximal independent set of $\{u_{i}:1\leq i\leq
n\}\cup\{u_{i}^{\prime}:1\leq i\leq n\}$ which contains $S$. For every $1\leq
i\leq n$, it holds that $|S^{\ast}\cap\{u_{i},u_{i}^{\prime}\}|=1$. Define
\[
\Phi:\{x_{1},\overline{x_{1}},...,x_{n},\overline{x_{n}}\}\longrightarrow
\{0,1\}
\]
by $\Phi(x_{i})=1\iff u_{i}\in S^{\ast}$ for every $1\leq i\leq n$. Since
$S^{\ast}$ dominates $\{v_{j}:1\leq j\leq m\}$, the function $\Phi$ satisfies
all clauses of $C$, and $I$ is a positive instance.
\end{proof}

\begin{figure}[h]
\setlength{\unitlength}{1.0cm} \begin{picture}(15,6)\thicklines
\put(6,1){\circle*{0.1}}
\put(6,0.7){\makebox(0,0){$y$}}
\multiput(3,2.5)(1.5,0){5}{\circle*{0.1}}
\put(6,1){\line(0,1){1.5}}
\put(6,1){\line(-2,1){3}}
\put(6,1){\line(-1,1){1.5}}
\put(6,1){\line(1,1){1.5}}
\put(6,1){\line(2,1){3}}
\multiput(3,4)(1.5,0){5}{\circle*{0.1}}
\multiput(3,2.5)(1.5,0){5}{\line(0,1){1.5}}
\put(2.7,2.5){\makebox(0,0){$a_{1}$}}
\put(4.2,2.5){\makebox(0,0){$a_{2}$}}
\put(5.7,2.5){\makebox(0,0){$a_{3}$}}
\put(7.2,2.5){\makebox(0,0){$a_{4}$}}
\put(8.7,2.5){\makebox(0,0){$a_{5}$}}
\multiput(2,5.5)(1.5,0){6}{\circle*{0.1}}
\put(2.7,4){\makebox(0,0){$v_{1}$}}
\put(4.2,4){\makebox(0,0){$v_{2}$}}
\put(5.7,4){\makebox(0,0){$v_{3}$}}
\put(7.2,4){\makebox(0,0){$v_{4}$}}
\put(8.7,4){\makebox(0,0){$v_{5}$}}
\multiput(2.5,5.5)(1.5,0){6}{\circle*{0.1}}
\multiput(2,5.5)(1.5,0){6}{\line(1,0){0.5}}
\put(1.9,5.8){\makebox(0,0){$u_{1}$}}
\put(3.4,5.8){\makebox(0,0){$u_{2}$}}
\put(4.9,5.8){\makebox(0,0){$u_{3}$}}
\put(6.4,5.8){\makebox(0,0){$u_{4}$}}
\put(7.9,5.8){\makebox(0,0){$u_{5}$}}
\put(9.4,5.8){\makebox(0,0){$u_{6}$}}
\put(2.6,5.8){\makebox(0,0){$u'_{1}$}}
\put(4.1,5.8){\makebox(0,0){$u'_{2}$}}
\put(5.6,5.8){\makebox(0,0){$u'_{3}$}}
\put(7.1,5.8){\makebox(0,0){$u'_{4}$}}
\put(8.6,5.8){\makebox(0,0){$u'_{5}$}}
\put(10.1,5.8){\makebox(0,0){$u'_{6}$}}
\put(3,4){\line(-2,3){1}}
\put(3,4){\line(2,3){1}}
\put(3,4){\line(4,3){2}}
\put(4.5,4){\line(-4,3){2}}
\put(4.5,4){\line(-2,3){1}}
\put(4.5,4){\line(4,3){2}}
\put(6,4){\line(2,3){1}}
\put(7.5,4){\line(2,3){1}}
\put(7.5,4){\line(5,3){2.5}}
\put(9,4){\line(-5,3){2.5}}
\put(9,4){\line(-2,3){1}}
\put(2,5.5){\circle*{0.25}}
\put(3.5,5.5){\circle*{0.25}}
\put(5.5,5.5){\circle*{0.25}}
\put(6.5,5.5){\circle*{0.25}}
\put(8.5,5.5){\circle*{0.25}}
\put(10,5.5){\circle*{0.25}}
\multiput(2.5,0.5)(0,-0.03){3}{\line(1,0){7}}
\multiput(2.5,3)(0,0.03){3}{\line(1,0){7}}
\multiput(2.5,0.43)(-0.03,0){3}{\line(0,1){2.64}}
\multiput(9.5,0.43)(0.03,0){3}{\line(0,1){2.64}}
\put(10,1.7){\makebox(0,0){$\bf{B}$}}
\begin{tikzpicture}[line width=0.8pt]
\draw (1,1) -- (1,1);
\draw (10.5,6.5) -- (7,5);
\draw (6.5,6.5) -- (10,5);
\draw (4.5,6.5) -- (8.5,5);
\draw (3,6.5) -- (7,5);
\end{tikzpicture}
\end{picture}
\caption{An example of the reduction from the \textbf{DSAT} problem to the
$\mathbf{GS}$ problem.}%
\label{gsc345}%
\end{figure}
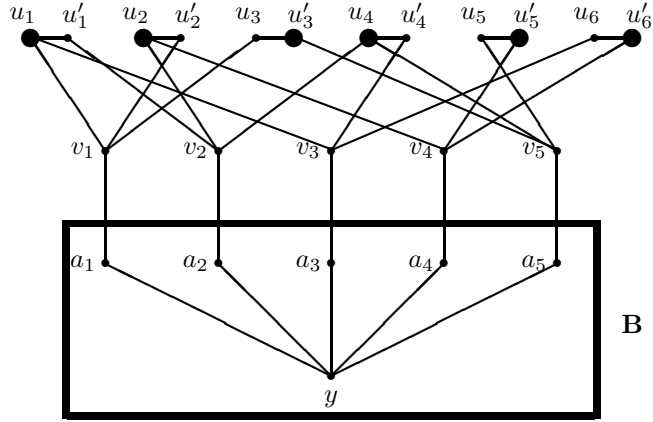

\begin{example}
\label{gsreduction} Let $I_{1} = (X,C)$ be an instance of the \textbf{DSAT}
problem, where $X =\{x_{1},x_{2},x_{3},x_{4},x_{5},x_{6}\}$ and $C =
\{\{x_{1},\overline{x_{2}},x_{3}\}, \{\overline{x_{1}},x_{2},x_{4}\},
\{x_{1},\overline{x_{4}},x_{6}\}, \newline\{x_{2},\overline{x_{5}}%
,\overline{x_{6}}\}, \{\overline{x_{3}},x_{4},x_{5}\}\}$. Then $I_{2}=(G,B)$
is an equivalent instance of the $\mathbf{GS}$ problem, where $G$ and $B$ are
the graphs shown in Figure \ref{gsc345}. The instance $I_{1}$ is positive
because of the satisfying assignment $\Phi$ defined by $\Phi(x_{i}) = 1$ if $i
\in\{1,2,4\}$, and $\Phi(x_{i}) = 0$ otherwise. The corresponding witness that
$I_{2}$ is positive is the set $\{u_{1},u_{2},u_{3}^{\prime},u_{4}%
,u_{5}^{\prime},u_{6}^{\prime}\}$.
\end{example}

\subsection{$K_{1,4}$-free graphs}

\begin{theorem}
\cite{cst:structures} \label{wck14npc} The following problem is \textbf{co}%
-\textbf{NP-}complete: \newline Input: A $K_{1,4}$-free graph $G$.\newline
Question: Is $G$ well-covered?
\end{theorem}

We use Theorem \ref{wck14npc} to prove the following.

\begin{theorem}
The $\mathbf{GS}$ problem is \textbf{NP-}complete even when its input is
restricted to $K_{1,4}$-free graphs.
\end{theorem}

\begin{proof}
Let $G$ be a $K_{1,4}$-free graph. An induced complete bipartite subgraph of
$G$ is isomorphic to $K_{i,j}$, for $1 \leq i \leq j \leq3$. Hence, the number
of these subgraphs is $O(n^{6})$, which is polynomial. Every unbalanced
induced complete bipartite subgraph of $G$ is a copy of $K_{1,2}$ or $K_{1,3}$
or $K_{2,3}$. The number of these subgraphs is $O(n^{5})$.

Assume, on the contrary, that there exists a polynomial algorithm solving the
$\mathbf{GS}$ problem for $K_{1,4}$-free graphs. The following algorithm
decides in polynomial time whether a $K_{1,4}$-free graph $G$ is well-covered.
For each induced complete bipartite unbalanced subgraph $B$ of $G$ on vertex
sets of bipartition $B_{X}$ and $B_{Y}$, decide in polynomial time whether $B$
is generating. Once an unbalanced generating subgraph is discovered, the
algorithm terminates announcing $G$ is not well-covered. If the algorithm
checked all induced complete bipartite unbalanced subgraphs of $G$, and none
of them is generating, then $G$ is well-covered. Hence, the $\mathbf{WC}$
problem can be solved in polynomial time when its input is restricted to
$K_{1,4}$-free graphs, but that contradicts Theorem \ref{wck14npc}. Thus the
$\mathbf{GS}$ problem is \textbf{NP-}complete, when its input is a $K_{1,4}%
$-free graph.
\end{proof}

\section{Polynomial algorithms when $\Delta$ is bounded}

In this section $G$ will be a graph with $n$ vertices and of maximum degree
$\Delta$. The main findings of this section are polynomial algorithms for the
$\mathbf{RE}$ problem and the $\mathbf{GS}$ problem in the restricted case,
when $\Delta$ is bounded. Our motivation here is the following.

\begin{theorem}
\cite{cer:degree} \label{deltacer} Let $k\in N$. The following problem is
polynomial.\newline Input: A graph $G$ with $\Delta_{G}\leq k\cdot(\log
_{2}n)^{\frac{1}{3}}$, and a function $w:V\left(  G\right)  \longrightarrow
\mathbb{R}$.\newline Question: Is $G$ $w$-well-covered?
\end{theorem}

We prove that the $\mathbf{GS}$ problem is polynomial, when $\Delta$ is
bounded using the same technique as in Theorem \ref{deltacer}.

\begin{theorem}
\label{deltagenerating} Let $k\in N$. The following problem can be solved in
$O(n^{2+2k^{3}})$ time.\newline Input: A graph $G$ such that $\Delta\leq
k\cdot(\log_{2}n)^{\frac{1}{3}}$, and an induced complete bipartite subgraph
$B$ of $G$.\newline Question: Is $B$ generating?
\end{theorem}

\begin{proof}
Let $B$ be an induced complete bipartite subgraph of $G$ on vertex sets of
bipartition $B_{X}$ and $B_{Y}$. For every $V\in\{X,Y\}$, let $U\in
\{X,Y\}-\{V\}$, and define:
\[
M_{1}(B_{V})=N(B_{V})\cap N_{2}(B_{U}),\ M_{2}(B_{V})=N(M_{1}(B_{V}))\cap
N_{2}(B_{V}).
\]
Then $\left\vert M_{1}(B_{V})\right\vert \leq k^{2}(\log_{2}n)^{2/3}$ and
$\left\vert M_{2}(B_{V})\right\vert \leq k^{3}\log_{2}n$. Obviously, $B$ is
generating if and only if there exists an independent set in $M_{2}(B_{X})\cup
M_{2}(B_{Y})$ that dominates $M_{1}(B_{X})\cup M_{1}(B_{Y})$.

The following algorithm decides whether $B$ is generating. For each subset $S$
of $M_{2}(B_{X})\cup M_{2}(B_{Y})$, check \ whether $S$ is \ independent \ and
\ dominates \ $M_{1}(B_{X})\cup M_{1}(B_{Y})$. Once an independent set
$S\subseteq M_{2}(B_{X})\cup M_{2}(B_{Y})$ is found such that $M_{1}%
(B_{X})\cup M_{1}(B_{Y})\subseteq N[S]$, the algorithm terminates announcing
the instance at hand is positive. If all subsets of $M_{2}(B_{X})\cup
M_{2}(B_{Y})$ were checked, and none of them is independent and dominates
$M_{1}(B_{X})\cup M_{1}(B_{Y})$, then the algorithm returns a negative answer.

The number of subsets the algorithm checks is
\[
O(2^{|M_{2}(B_{X})\cup M_{2}(B_{Y})|})=O(2^{2k^{3}\log_{2}n})=O(n^{2k^{3}}).
\]
For each subset $S$, the decision whether $S$ is both independent and
dominates $M_{1}(B_{X})\cup M_{1}(B_{Y})$ can be done in $O(n^{2})$.
Therefore, the algorithm terminates in $O(n^{2+2k^{3}})$ time, which is polynomial.
\end{proof}

We next prove that the $\mathbf{RE}$ problem is polynomial for the less
restrictable bound in comparison with its counterpart from Theorem
\ref{deltagenerating}.

\begin{theorem}
Let $k\in N$. The following problem can be solved in $O(n^{2+2k^{2}})$
time.\newline Input: A graph $G$ such that $\Delta\leq k\cdot(\log
_{2}n)^{\frac{1}{2}}$, and an edge $xy \in E$.\newline Question: Is $xy$ relating?
\end{theorem}

\begin{proof}
For every $v\in\{x,y\}$, let $u\in\{x,y\}-\{v\}$. Define: $M_{1}(v)=N(v)\cap
N_{2}(u)$, $M_{2}(v)=N(M_{1}(v))\cap N_{2}(v)$. Then $\left\vert
M_{1}(v)\right\vert \leq k\cdot(\log_{2}n)^{\frac{1}{2}}$ and $\left\vert
M_{2}(v)\right\vert \leq k^{2}\log_{2}n$. Clearly, $xy$ is relating if and
only if there exists an independent set in $M_{2}(x)\cup M_{2}(y)$, which
dominates $M_{1}(x)\cup M_{1}(y)$.

The following algorithm decides whether $xy$ is relating. For each subset $S$
of $M_{2}(x)\cup M_{2}(y)$, check whether $S$ is independent and dominates
$M_{1}(x)\cup M_{1}(y)$. Once \ an \ independent \ set \ \ $S\subseteq
M_{2}(x)\cup M_{2}(y)$ \ \ is \ found \ such \ that \ \ $M_{1}(x)\cup
M_{1}(y)\subseteq N[S]$, the algorithm terminates announcing the instance at
hand is positive. If all subsets of $M_{2}(x)\cup M_{2}(y)$ were checked, and
none of them is both independent and dominates $M_{1}(x)\cup M_{1}(y)$, then
the algorithm returns a negative answer.

The number of subsets the algorithm checks is
\[
O(2^{|M_{2}(x)\cup M_{2}(y)|})=O(2^{2k^{2}\log_{2}n})=O(n^{2k^{2}}).
\]
For each subset $S$, the decision whether $S$ is both independent and
dominates $M_{1}(x)\cup M_{1}(y)$ can be done in $O(n^{2})$ time. Therefore,
the algorithm terminates in $O(n^{2+2k^{2}})$ time.
\end{proof}

In what follows, our purpose is both to formalize and to give a detailed proof
of a claim mentioned in \cite{cer:degree}.

\begin{theorem}
\label{deltawcw} Let $k\in N$. The following problem can be solved in
$O(n^{3+2k^{2}+2k^{3}})$ time.\newline Input: A graph $G$ such that
$\Delta\leq k\cdot(\log_{2}n)^{\frac{1}{3}}$.\newline Output: The vector space
$WCW(G)$.
\end{theorem}

\begin{proof}
Let $G$ be a graph such that $\Delta\leq k\cdot(\log_{2}n)^{\frac{1}{3}}$. For
every vertex $v\in V$, let $L_{v}$ be the vector space of all weight functions
$w:V\left(  G\right)  \longrightarrow\mathbb{R}$ which satisfy all
restrictions of all generating subgraphs which contain the vertex $v$.
Clearly, $WCW(G)={\bigcap\limits_{v\in V\left(  G\right)  }}L_{v}$. Hence, we
first present an algorithm for finding $L_{v}$ for every $v\in V$.

Let $v\in V$. Since the diameter of every complete bipartite graph is at most
$2$, every complete bipartite subgraph of $G$ which contains $v$ is a subgraph
of $N_{2}[v]$. However,
\[
\left\vert N_{2}(v)\right\vert \leq\Delta^{2}\leq k^{2}(\log_{2}n)^{\frac
{2}{3}},
\]
and
\[
\left\vert N_{2}[v]\right\vert \leq2\left\vert N_{2}(v)\right\vert \leq
2k^{2}(\log_{2}n)^{\frac{2}{3}}.
\]
Therefore, the number of induced complete bipartite subgraphs which contain
$v$ cannot exceed
\[
2^{2k^{2}(\log_{2}n)^{\frac{2}{3}}}\leq n^{2k^{2}}.
\]

The following algorithm finds $L_{v}$:

\begin{itemize}
\item For each induced complete bipartite subgraph $B=(B_{X},B_{Y})$ of $G$
containing $v$:

\begin{itemize}
\item Decide whether $B$ is generating;

\item If $B$ is generating add the restriction $w(B_{X})=w(B_{Y})$ to the list
of equations defining $L_{v}$.
\end{itemize}
\end{itemize}

We have proved that the number of induced complete bipartite subgraphs of $G$
containing $v$ cannot exceed $n^{2k^{2}}$. By Theorem \ref{deltagenerating},
deciding for each subgraph whether it is generating can be done in
$O(n^{2+2k^{3}})$ time. Therefore, the algorithm for finding $L_{v}$
terminates in $O(n^{2+2k^{2}+2k^{3}})$ time. In order to find $WCW(G)$, the
algorithm for finding $L_{v}$ should be invoked $n$ times. Therefore, finding
$WCW(G)$ can be completed in $O(n^{3+2k^{2}+2k^{3}})$ time.
\end{proof}

\section{Conclusions and future work}

The following table presents complexity results concerning the four major
problems presented in this paper. The empty table cells correspond to unsolved
cases. In particular, we want to find the complexity status of the
$\mathbf{WCW}$ problem for bipartite graphs and for graphs with girth $6$ at
least. For these families of graphs the $\mathbf{GS}$ problem is
\textbf{NP-}complete while the $\mathbf{WC}$ problem is polynomial. Hence,
either we obtain a family of graphs for which the $\mathbf{WC}$ problem is
polynomial while the $\mathbf{WCW}$ problem is \textbf{co}-\textbf{NP-}hard,
or we obtain a family of graphs for which the $\mathbf{GS}$ problem is
\textbf{NP-}complete while the $\mathbf{WCW}$ problem is polynomial.

In addition, we are interested in finding some polynomial relaxations of the
bipartite case, if any. For instance, can recognizing well-covered graphs
belonging to $\mathcal{G}(\widehat{C_{3}},\widehat{C_{5}})$ be done polynomially?

Let us emphasize that we do not know whether there exists a family of graphs
for which the $\mathbf{RE}$ problem can be solved in polynomial time, but the
$\mathbf{GS}$ problem is \textbf{NP-}complete.

\begin{center}
\begin{table}[ptb]%
\begin{tabular}
[c]{|c|c|c|c|c|}\hline
\textbf{{Input}} & $\mathbf{WC}$ & $\mathbf{{WCW}}$ & $\mathbf{RE}$ &
$\mathbf{GS}$\\\hline%
\begin{tabular}
[c]{c}%
\\
general
\end{tabular}
&
\begin{tabular}
[c]{c}%
\textbf{co}-\textbf{NPC}\\
\cite{cs:note,sknryn:compwc}%
\end{tabular}
&
\begin{tabular}
[c]{c}%
\textbf{co}-\textbf{NPH}\\
\cite{cs:note,sknryn:compwc}%
\end{tabular}
&
\begin{tabular}
[c]{c}%
\textbf{NPC}\\
\cite{bnz:wcc4}%
\end{tabular}
&
\begin{tabular}
[c]{c}%
\textbf{NPC}\\
\cite{bnz:wcc4}%
\end{tabular}
\\\hline%
\begin{tabular}
[c]{c}%
\\
$K_{1,3}$-free
\end{tabular}
&
\begin{tabular}
[c]{c}%
\textbf{P}\\
\cite{tata:wck13f}%
\end{tabular}
&
\begin{tabular}
[c]{c}%
\textbf{P}\\
\cite{lt:equimatchable}%
\end{tabular}
&
\begin{tabular}
[c]{c}%
\textbf{P}\\
\cite{tata:wck13fn}%
\end{tabular}
&
\begin{tabular}
[c]{c}%
\textbf{P}\\
\cite{tata:wck13fn}%
\end{tabular}
\\\hline%
\begin{tabular}
[c]{c}%
\\
$K_{1,4}$-free
\end{tabular}
&
\begin{tabular}
[c]{c}%
\textbf{co}-\textbf{NPC}\\
\cite{cst:structures}%
\end{tabular}
&
\begin{tabular}
[c]{c}%
\textbf{co}-\textbf{NPH}\\
\cite{cst:structures}%
\end{tabular}
&
\begin{tabular}
[c]{c}%
\\
\end{tabular}
&
\begin{tabular}
[c]{c}%
\textbf{NPC}\\
this paper
\end{tabular}
\\\hline%
\begin{tabular}
[c]{c}%
\\
$\mathcal{G}(\widehat{C_{4}},\widehat{C_{5}})$%
\end{tabular}
&
\begin{tabular}
[c]{c}%
\textbf{P}\\
\cite{fhn:wc45}%
\end{tabular}
&  &
\begin{tabular}
[c]{c}%
\textbf{NPC}\\
\cite{lt:relatedc4}%
\end{tabular}
&
\begin{tabular}
[c]{c}%
\textbf{NPC}\\
\cite{lt:relatedc4}%
\end{tabular}
\\\hline%
\begin{tabular}
[c]{c}%
\\
$\mathcal{G}(\widehat{C_{4}},\widehat{C_{6}})$%
\end{tabular}
&
\begin{tabular}
[c]{c}%
\\
\end{tabular}
&
\begin{tabular}
[c]{c}%
\\
\end{tabular}
&
\begin{tabular}
[c]{c}%
\textbf{P}\\
\cite{lt:relatedc4}%
\end{tabular}
&
\begin{tabular}
[c]{c}%
\\
\end{tabular}
\\\hline%
\begin{tabular}
[c]{c}%
\\
$\mathcal{G}(\widehat{C_{5}},\widehat{C_{6}})$%
\end{tabular}
&
\begin{tabular}
[c]{c}%
\\
\end{tabular}
&
\begin{tabular}
[c]{c}%
\\
\end{tabular}
&
\begin{tabular}
[c]{c}%
\textbf{P}\\
\cite{lt:wwc456}%
\end{tabular}
&
\begin{tabular}
[c]{c}%
\\
\end{tabular}
\\\hline%
\begin{tabular}
[c]{c}%
\\
$\mathcal{G}(\widehat{C_{5}},\widehat{C_{6}},\widehat{C_{7}})$%
\end{tabular}
&
\begin{tabular}
[c]{c}%
\\
\end{tabular}
&
\begin{tabular}
[c]{c}%
\\
\end{tabular}
&
\begin{tabular}
[c]{c}%
\textbf{P}\\
\cite{lt:wwc456}%
\end{tabular}
&
\begin{tabular}
[c]{c}%
\textbf{P}\\
\cite{lt:wwc456}%
\end{tabular}
\\\hline%
\begin{tabular}
[c]{c}%
\\
$\mathcal{G}(\widehat{C_{4}},\widehat{C_{5}},\widehat{C_{6}})$%
\end{tabular}
&
\begin{tabular}
[c]{c}%
\textbf{P}\\
\cite{fhn:wc45}%
\end{tabular}
&
\begin{tabular}
[c]{c}%
\textbf{P}\\
\cite{lt:wwc456}%
\end{tabular}
&
\begin{tabular}
[c]{c}%
\textbf{P}\\
\cite{lt:wwc456}%
\end{tabular}
&
\begin{tabular}
[c]{c}%
\textbf{P}\\
\cite{lt:wwc456}%
\end{tabular}
\\\hline%
\begin{tabular}
[c]{c}%
\\
$\mathcal{G}(\widehat{C_{4}},\widehat{C_{6}},\widehat{C_{7}})$%
\end{tabular}
&
\begin{tabular}
[c]{c}%
\\
\end{tabular}
&
\begin{tabular}
[c]{c}%
\\
\end{tabular}
&
\begin{tabular}
[c]{c}%
\textbf{P}\\
\cite{lt:wc4567}%
\end{tabular}
&
\begin{tabular}
[c]{c}%
\textbf{P}\\
\cite{lt:wc4567}%
\end{tabular}
\\\hline%
\begin{tabular}
[c]{c}%
\\
bipartite
\end{tabular}
&
\begin{tabular}
[c]{c}%
\textbf{P}\\
\cite{ravindra:well-covered}%
\end{tabular}
&  &
\begin{tabular}
[c]{c}%
\textbf{NPC}\\
this paper
\end{tabular}
&
\begin{tabular}
[c]{c}%
\textbf{NPC}\\
this paper
\end{tabular}
\\\hline%
\begin{tabular}
[c]{c}%
\\
$\mathcal{G}(\widehat{C_{3}},\widehat{C_{4}})$%
\end{tabular}
&
\begin{tabular}
[c]{c}%
\textbf{P}\\
\cite{fhn:wcg5}%
\end{tabular}
&  &
\begin{tabular}
[c]{c}%
\\
\end{tabular}
&
\begin{tabular}
[c]{c}%
\textbf{NPC}\\
this paper
\end{tabular}
\\\hline%
\begin{tabular}
[c]{c}%
\\
$\mathcal{G}(\widehat{C_{3}},\widehat{C_{4}},\widehat{C_{5}})$%
\end{tabular}
&
\begin{tabular}
[c]{c}%
\textbf{P}\\
\cite{fhn:wcg5}%
\end{tabular}
&  &
\begin{tabular}
[c]{c}%
\\
\end{tabular}
&
\begin{tabular}
[c]{c}%
\textbf{NPC}\\
this paper
\end{tabular}
\\\hline%
\begin{tabular}
[c]{c}%
\\
$\Delta\leq k(\log_{2}n)^{\frac{1}{3}}$%
\end{tabular}
&
\begin{tabular}
[c]{c}%
\textbf{P}\\
\cite{cer:degree}%
\end{tabular}
&
\begin{tabular}
[c]{c}%
\textbf{P}\\
\cite{cer:degree}%
\end{tabular}
&
\begin{tabular}
[c]{c}%
\textbf{P}\\
\cite{cer:degree} and this paper
\end{tabular}
&
\begin{tabular}
[c]{c}%
\textbf{P}\\
\cite{cer:degree} and this paper
\end{tabular}
\\\hline%
\begin{tabular}
[c]{c}%
\\
$\Delta\leq k(\log_{2}n)^{\frac{1}{2}}$%
\end{tabular}
&
\begin{tabular}
[c]{c}%
\\
\end{tabular}
&
\begin{tabular}
[c]{c}%
\\
\end{tabular}
&
\begin{tabular}
[c]{c}%
\textbf{P}\\
\cite{cer:degree} and this paper
\end{tabular}
&
\begin{tabular}
[c]{c}%
\\
\end{tabular}
\\\hline
\end{tabular}
\caption{Complexity results on the $4$ problems.}%
\end{table}
\end{center}

Another interesting open question is whether there exists a family of graphs
for which the $\mathbf{GS}$ problem is polynomial and its corresponding
$\mathbf{WCW}$ problem is \textbf{co}-\textbf{NP-}hard.

\begin{figure}[h]
\setlength{\unitlength}{1.0cm} \begin{picture}(15,5)\thicklines
\put(3,2.5){\circle*{0.1}}
\put(3,2.8){\makebox(0,0){$v_{1}$}}
\multiput(4.5,1)(1.5,0){3}{\circle*{0.1}}
\multiput(4.5,4)(1.5,0){3}{\circle*{0.1}}
\put(4.5,0.7){\makebox(0,0){$v_{2}$}}
\put(6,0.7){\makebox(0,0){$v_{4}$}}
\put(7.5,0.7){\makebox(0,0){$v_{6}$}}
\put(4.5,4.3){\makebox(0,0){$v_{3}$}}
\put(6,4.3){\makebox(0,0){$v_{5}$}}
\put(7.5,4.3){\makebox(0,0){$v_{7}$}}
\put(3,2.5){\line(1,-1){1.5}}
\put(3,2.5){\line(1,1){1.5}}
\put(4.5,1){\line(1,0){3}}
\put(4.5,4){\line(1,0){3}}
\put(7.5,1){\line(0,1){3}}
\put(7,2.5){\circle*{0.1}}
\put(7,2.8){\makebox(0,0){$v_{8}$}}
\put(9,3){\circle*{0.1}}
\put(9,2){\circle*{0.1}}
\put(9.4,3){\makebox(0,0){$v_{9}$}}
\put(9.4,2){\makebox(0,0){$v_{10}$}}
\put(7.5,1){\line(-1,3){0.5}}
\put(7,2.5){\line(-2,3){1}}
\put(7.5,1){\line(3,4){1.5}}
\put(7,2.5){\line(4,-1){2}}
\put(7.5,1){\line(3,2){1.5}}
\put(7.5,4){\line(3,-2){1.5}}
\multiput(2.43,0.5)(0,-0.02){4}{\line(1,0){3.15}}
\multiput(2.43,4.5)(0,0.02){4}{\line(1,0){3.15}}
\multiput(2.5,0.5)(-0.02,0){4}{\line(0,1){4}}
\multiput(5.5,0.5)(0.02,0){4}{\line(0,1){4}}
\put(4.3,2.5){\makebox(0,0){$\bf{B}$}}
\end{picture}
\caption{The failure of the naive algorithm.}%
\label{Fig2}%
\end{figure}
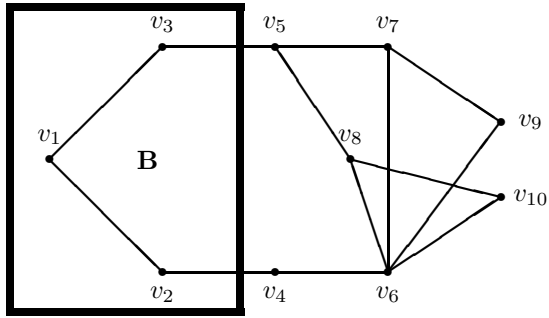

The naive algorithm for the $\mathbf{GS}$ problem, receives as its input an
instance $I=(G,B=(B_{X},B_{Y}))$. Then it finds $WCW(G)$. If there exists a
weight function $w\in WCW(G)$ such that $w(B_{X})\neq w(B_{Y})$, then $B$ is
not generating, and consequently, $I$ is negative. Otherwise, $I$ is positive.
For every family $\Psi$ of graphs, if the $\mathbf{WCW}$ problem can be solved
polynomially, then the naive algorithm for the $\mathbf{GS}$ problem
terminates polynomially.

However, the naive algorithm fails, when its input is $(G,B)$, where $G$ is
the graph shown in Figure \ref{Fig2}, and $B$ is the subgraph induced by
$\{v_{1},v_{2},v_{3}\}$. A function $w:V\left(  G\right)  \longrightarrow
\mathbb{R}$ belongs to $WCW(G)$ if and only if the following conditions hold:

\begin{itemize}
\item $w(v_{7})=w(v_{9})$

\item $w(v_{8})=w(v_{10})$

\item $w(v_{6})=w(v_{9})+w(v_{10})$

\item $w(v_{i}) = 0$ for every $1 \leq i \leq5$.
\end{itemize}

Hence, $w(v_{1})=w(v_{2})+w(v_{3})$ for every $w\in WCW(G)$, and the naive
algorithm decides that $B$ is generating, although it is not.

\end{document}